\documentclass[9pt]{article}

\usepackage{amsfonts,amssymb,amsmath}
\usepackage{amsthm}
\usepackage{graphicx}
\usepackage[ruled,vlined,commentsnumbered]{algorithm2e}
\usepackage[usenames,dvipsnames]{color} \usepackage{hyperref}
\usepackage{multirow} \usepackage{lineno}
\usepackage[shortlabels]{enumitem} \usepackage[utf8]{inputenc}
\usepackage[OT4]{fontenc} \usepackage{comment} \usepackage{fancyhdr}
\usepackage{wrapfig}
 \newcommand{\Problem}{\textsc{EMB}}
\newcommand{\ProblemT}{\textsc{EMB(2)}}

\newcommand{\variab}{\nu}

\newcommand{\ChunkType}{\tau}
\newcommand{\VirtualNodes}{\ensuremath{V}}

\newcommand{\aroot}{\emph{root}}

\newcommand{\clauses}{\alpha} \newcommand{\vars}{\beta}
\newcommand{\variables}{\beta} \newcommand{\achunk}{\ensuremath{c}}
 \newcommand{\capa}{\emph{cap}}
\newcommand{\capacity}{\emph{cap}}

\newcommand{\dist}{\emph{dist}}

 \newcommand{\Cost}{\textsc{F}}

\newcommand{\Tree}{\ensuremath{T}}
\newcommand{\CostTrans}{\ensuremath{b_1}}
\newcommand{\CostCom}{\ensuremath{b_2}}
\newcommand{\Vms}{\ensuremath{n_V}} 
 \newcommand{\TSAT}{\textsc{3-Sat}}
 \newcommand{\SAT}{\textsc{Sat}}
\newcommand{\ZSAT}{\textsc{2-Sat}}

\newcommand{\Formula}{\ensuremath{\Psi}}

\newcommand{\Thr}{\ensuremath{Th}}

\newcommand{\positive}{\ensuremath{positive}}
\newcommand{\negative}{\ensuremath{negative}}
\newcommand{\Val}{\ensuremath{Val}}
\newcommand{\Sol}{\ensuremath{SOL}}

\definecolor{blueLink}{rgb}{0,0.2,0.8}
\hypersetup{colorlinks,linkcolor=blueLink,urlcolor=blueLink,citecolor=blueLink}
\newcommand{\lref}[2][]{\hyperref[#2]{#1~\ref*{#2}}}

\newtheorem{theorem}{Theorem}

\newtheorem{lemma}[theorem]{Lemma}

\definecolor{brown}{rgb}{0.4,0,0} 
\definecolor{purple}{rgb}{0.2,0,0.6}
\definecolor{hotpink}{rgb}{1,0.4,0.7}

\title{Hardness of Virtual Network Embedding with Replica Selection}

\author{Carlo Fuerst, Maciej Pacut, Stefan Schmid }

\date{}
\begin{document}

\maketitle

\sloppy

\begin{abstract}

Efficient embedding virtual clusters in physical network is a
challenging problem. In this paper we consider a scenario where
physical network has a structure of a balanced tree. This assumption
is justified by many real-world implementations of datacenters.

We consider an extension to virtual cluster embedding by introducing
replication among data chunks. In many real-world applications, data
is stored in distributed and redundant way. This assumption introduces
additional hardness in deciding what replica to process.

By reduction from classical NP-complete problem of Boolean
Satisfiability, we show limits of optimality of embedding. Our result
holds even in trees of edge height bounded by three. Also, we show
that limiting replication factor to two replicas per chunk type does
not make the problem simpler.
\end{abstract}

\section{Introduction}

Server virtualization has revamped the server business over the last
years, and has radically changed the way we think about resource
allocation: today, almost arbitrary computational resources can be
allocated on demand.  Moreover, the virtualization trend now started
to spill over to the network: batch-processing applications such as
MapReduce often generate significant network traffic (namely during
the so-called shuffle phase)~\cite{amazonbw}, and in order to avoid
interference in the underlying physical network and in order to
provide a predictable application performance, it is important to
provide performance isolation and bandwidth guarantees for the virtual
network connecting the virtual machines.~\cite{talk-about}

Prominent example of large-scale framework that is used in
datacenters is MapReduce. In such applications often network usage becomes
the limiting factor for application performance. Sharing single link
among multiple node-node communications requires reserving certain
amount of network traffic to avoid slowing the transfer
down. In order to model MapReduce execution, bandwidth has to be
reserved for transfer of chunks to nodes and for node-node
interconnection. Resulting virtual network to be embedded consists of clique and
single links incoming to its vertices. Described abstraction is called
\emph{virtual
  cluster}~\cite{oktopus,proteus}.

\textbf{Our Contributions.}We show that minimizing network footprint is NP-hard in presence of
multiple replicas of the same chunk type.  Moreover, we show that
NP-hard problems already arise in small-diameter networks (as they are
widely used today~\cite{fattree}), and even if the number of replicas
is bounded by two.

\section{Model}\label{sec:model}

To get started, and before introducing our formal model and its
constituting parts in detail, we will discuss the practical
motivation.

\subsection{Background and Practical Motivation}

Our model is motivated by batch-processing applications such as
MapReduce.  Such applications use multiple virtual machines to process
data, initially often redundantly stored in a distributed file system
implemented by multiple servers.~\cite{mapreduce} The standard
datacenter topologies today are (multi-rooted) fat-tree
resp.~\emph{Clos} topologies~\cite{fattree,vl2}, hierarchical networks
recursively made of sub-trees at each level; servers are located at
the tree leaves. Given the amount of multiplexing over the mesh of
links and the availability of multi-path routing protocol, e.g.~ECMP,
the redundant links can be considered as a single aggregate link for
bandwidth reservations~\cite{oktopus,proteus}.

During execution, batch-processing applications typically cycle
through different phases, most prominently, a mapping phase and a
reducing phase; between the two phases, a shuffling operation is
performed, a phase where the results from the mappers are communicated
to the reducers. Since the shuffling phase can constitute a
non-negligible part of the overall runtime~\cite{orchestra}, and since
concurrent network transmissions can introduce interference and
performance unpredictability~\cite{amazonbw}, it is important to
provide explicit minimal bandwidth guarantees~\cite{talk-about}.  In
particular, we model the virtual network connecting the virtual
machines as a virtual cluster~\cite{oktopus,talk-about,proteus};
however, we extend this model with a notion of data-locality.  In
particular, we distinguish between the bandwidth needed between
assigned chunk and virtual machine ($\CostTrans$) and the bandwidth
needed between two virtual machines ($\CostCom$); in practice, for
applications with a large ``mapping ratio'' where the mapping phase
already reduces the data size significantly, it may hold that
$\CostCom\ll\CostTrans$.

\subsection{Fundamental Parts}

Let us now introduce our model more formally. It consists of three
fundamental parts: (1) the substrate network (the servers and the
connecting physical network), (2) the to be processed input (the data
chunks), and (3) the virtual network (the virtual machines and the
logical network connecting the machines to each other as well as to
the chunks).

\textbf{\emph{The Substrate Network.}} The substrate network (also
known as the \emph{host graph}) represents the physical resources: a
set $S$ of $n_S=|S|$ servers interconnected by a network consisting of
a set $R$ of routers (or switches) and a set $E$ of (symmetric) links;
we will often refer to the elements in $S\cup R$ as the
\emph{vertices}. We will assume that the inter-connecting network
forms an (arbitrary, not necessarily balanced or regular) tree, where
the servers are located at the tree leaves.  Each server $s\in S$ can
host zero or one virtual machine. Each link $e\in E$ has a certain
bandwidth capacity $\capacity(e)$.

\textbf{\emph{The Input Data.}} The to be processed data constitutes
the input to the batch-processing application.  The data is stored in
a distributed manner; this spatial distribution is given and not
subject to optimization.  The input data consists of $\tau$ different
\emph{chunk types} $\{\achunk_1, \ldots, \achunk_{\ChunkType}\}$,
where each chunk type $\achunk_i$ can have $r_i\geq 1$ instances (or
replicas) $\{\achunk_{i}^{(1)},\ldots, \achunk_{i}^{(r_i)}\}$, stored
at different servers.  It is sufficient to process one replica, and we
will sometimes refer to this replica as the \emph{active} (or
selected) replica.

The input data is stored redundantly, and the algorithm has the
freedom to choose a replica for each chunk type, and assign it to a
virtual machine (i.e., \emph{node}).

\textbf{\emph{The Virtual Network.}} The virtual network consists of a
set $\VirtualNodes$ of $n_V=|\VirtualNodes|$ virtual machines,
henceforth often simply called \emph{nodes}.  Each node
$v \in \VirtualNodes$ can be placed (or, synonymously,
\emph{embedded}) on a server; this placement can be subject to
optimization.

Please note that number of nodes might exceed the number of chunk
types. Excessive machines (or idle) do not process chunks, but
participate in shuffle and reduce phase of Map-Reduce. Excessive
machines do not have matched chunks, therefore their transportation
cost is zero. Every machine, idle or not -- incline communication cost
to other machines.

We will denote the server $s$ hosting node $v$ by $\pi(v)=s$.  Since
these nodes process the input data, they need to be assigned and
connected to the chunks. Concretely, for each chunk type $\achunk_i$,
exactly one replica $\achunk_{i}^{(j)}$ must be processed by exactly
one node $v$; which replica $\achunk_{i}^{(k)}$ is chosen is subject
to optimization, and we will denote by $\mu$ the assignment of nodes
to chunks.

In order to ensure a predictable application performance, both the
connection to the chunks as well as the interconnection between the
nodes may have to ensure certain minimal bandwidth guarantees; we will
refer to the first type of virtual network as the \emph{(chunk) access
  network}, and to the second type of virtual network as the
\emph{(node) inter-connect}; the latter is modeled as a complete
network (a \emph{clique}). Concretely, we assume that an active chunk
is connected to its node at a minimal (guaranteed) bandwidth
$\CostTrans$, and a node is connected to any other node at minimal
(guaranteed) bandwidth $\CostCom$.

We allow the number of chunk types to be smaller
than the number of nodes.  The ``idle'' nodes however do participate
in the inter-connect communication (in practical terms: in the shuffle
phase and the reducing phase).

\subsection{Optimization Objective}

Our goal is to develop algorithms which minimize the \emph{resource
  footprint}: the guaranteed bandwidth allocation (or synonymously:
\emph{reservation}) on all links of the given embedding; note that
only the resource allocation at the links but not at the servers
depends on the replica selection or embedding. Thus, we on the one
hand aim to embed the nodes in a locality-aware manner, close to the
input data (the chunks), but at the same time also aim to embed the
nodes as close as possible to each other.

Formally, let $\dist(v,\achunk)$ denote the distance (in the
underlying physical network $\Tree$) between a node $v$ and its
assigned (active) chunk replica $\achunk$, and let $\dist(v_1,v_2)$
denote the distance between the two nodes $v_1$ and $v_2$.  We define
the \emph{footprint} $\Cost(v)$ of a node $v$ as follows:
$$
\Cost(v) = \underbrace{\sum_{\achunk\in \mu(v)} \CostTrans \cdot
  \dist(v,\achunk)}_{\text{transportation}} + \underbrace{\frac{1}{2}
  \cdot \sum_{v' \in \VirtualNodes\setminus\{v\}} \CostCom \cdot
  \dist(v,v')}_{\text{inter-connect}},
$$
\noindent where $\mu(v)$ is the set of chunks assigned to $v$. Our
goal is to minimize the overall footprint
$\Cost=\sum_{v\in V} \Cost(v)$.
% Recall Figure~\ref{fig:overview}.

% \textbf{Problem Decomposition.}  \textbf{Remark on Practical
% Motivation.}

\subsection{Decision problem}

In order to perform a reduction from NP-complete problem, we need to
transform our optimization problem to decision problem. To do so, we
define $\Problem$ as a set containing pair $\lbrace k, I \rbrace$,
iff. $I$ is an instance of virtual cluster embedding problem that has
feasible (bandwidth-respecting) solution of cost $\leq k$.

%%%%%%%%%%%%%%%%%%%%%%%%%%%%%%%%%%%%%
\section{Hardness of problem with multiple replicas allowed}\label{sec:rep}

We prove that $\Problem$ is NP-hard by reduction from the
\emph{Boolean Satisfiability Problem} ($\SAT$).  Since $\SAT$ is a
decision problem, we introduce a cost threshold $\Thr$ to transform
$\Problem$ into a decision problem too.

Let's first recall that the $\SAT$ problem asks whether a positive
valuation exists for a formula $\Formula$ with $\clauses$ clauses and
$\variables$ variables.  In the following, we will only focus on
$\SAT$ instances of at least four variables; this $\SAT$ variant is
still NP-hard.

\textbf{Construction.}  Given any formula $\Formula$ in
\emph{Conjunctive Normal Form (CNF)} with $\clauses$ clauses and
$\variables \geq 4$ variables, we produce a $\Problem$ instance as
follows: First, we construct a substrate tree $\Tree_{\Formula}$,
consisting of a root and separate gadgets for each variable of
$\Formula$, each of which is a child of the root.  The gadget of
variable $\variab$ consists of $\aroot(\variab)$ and its two children:
$\positive(\variab)$ and $\negative(\variab)$. Child
$\positive(\variab)$ has $\clauses$ many children labeled
$\nu_1, \nu_2, \ldots, \nu_{\clauses}$, and child $\negative(\variab)$
has $\clauses$ many children labeled
$\neg \nu_1, \neg \nu_2, \dots, \neg \nu_{\clauses}$. Every gadget has
the same structure: the same height and the same number of
leaves. This construction is illustrated in
Table~\ref{fig:construction_3sat}.

\begin{figure}
  \includegraphics{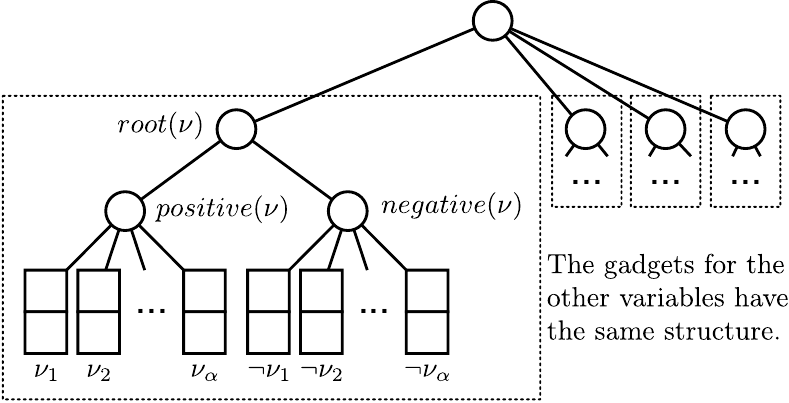}
  \caption{The construction of the variable gadget for $\nu$. If $\nu$
    appears in the first clause, a chunk $\achunk_1$, will be located
    at $\nu_1$. If $\neg \nu$ satisfies the last clause,
    $\neg \nu_\alpha$ will host $\achunk_\alpha$.}
  \label{fig:construction_3sat}
\end{figure}

For all variables $\nu$, we set the bandwidth on the link, connecting
$root(\nu)$ to the root of the substrate tree, to
$\clauses\cdot(\clauses\cdot\variables-\clauses)$. The bandwidth on
the other edges is not limited.

We set the number of nodes to $\Vms = \clauses \cdot \variables$.
Moreover, we define the inter-connect communication cost to be $1$,
and the access cost to be a sufficiently large constant $W$, such that
nodes must always be collocated with chunks ($W = \Thr + 1$ is
sufficient).

% (For a concrete value, see Appendix~\ref{ap:W}).

We set the number of chunk types to be equal to the number of clauses,
$\tau = \clauses$. To finish our construction, we place data chunks at
leaves, as follows: for the $i$-th clause we construct as many
replicas of chunk $\achunk_i$ as there are literals in the clause. For
each literal $\ell$ (of the form $\variab$ or $\neg \variab$) that
satisfies clause $i$, we place a replica of chunk $\achunk_i$ in the
leaf labeled $\ell_i$.

Note, that in this construction some nodes will be idle. No chunks
will be assigned to these nodes, but they will nevertheless
participate in the node interconnect.

We set the threshold $\Thr$ to:
% (\maciek{we might not want to introduce those variable names, as
% those have no meaning here}):
$ \Thr = \variables \cdot ({\clauses \choose 2} \cdot 2 + \clauses
\cdot (\clauses \cdot \variables - \clauses))$.

\textbf{Proof of correctness of construction.}  We now show that our
construction indeed decides $\SAT$. We set the capacities such that in
every gadget, at most $\clauses$ nodes can be mapped, where $\clauses$
is the number of clauses of $\Formula$ .  We can apply the Bandwidth
Lemma (Lemma~\ref{lem:bandwidth-lemma}) as follows: We interpret $a_i$
as the number of nodes that are embedded in the $i$-th gadget,
$\clauses$ as the number of clauses, and $\variables$ as the number of
variables.  The LHS of the inequality of
Lemma~\ref{lem:bandwidth-lemma} is a formula for the communication
cost of nodes inside the $i$-th gadget to nodes outside the
gadget. The RHS of the inequality is the bandwidth constraint for the
gadget. This implies that any feasible solution must embed exactly
$\clauses$ nodes in every gadget.  Recall that in our $\SAT$ instance,
we have at least four variables.

\begin{theorem}
  The problem $\Problem$ is NP-hard.
  \label{theorem:3sat}
\end{theorem}
\begin{proof}
  We will prove that formula $\Formula$ is satisfiable iff $\Problem$
  has a solution of cost $\leq \Thr$.

  ($\Rightarrow$) Let us take any valuation $\Val$ that satisfies
  $\Formula$.  We will construct a solution to $\Problem$ using $\Val$
  in the following way.  For each variable $\variab$ in $\Formula$, we
  embed $\clauses$ many nodes at the leaves of the gadget of
  $\variab$. We need to choose $\clauses$ out of $2 \cdot \clauses$
  leaves to embed nodes. If $\Val(\variab) = 1$, we embed nodes at the
  leaves of $\positive(\variab)$, else we embed all nodes at leaves
  $\negative(\variab)$.  The solution constructed this way has cost
  exactly $\Thr$, because the nodes are evenly split among gadgets,
  and nodes are not distributed across $\positive(\variab)$ and
  $\negative(\variab)$ subtrees.

  We calculate the chunk-node matching $\mu$ by assigning every chunk
  to the node which is collocated with the first chunk replica. This
  solution is feasible because every clause of $\Psi$ was satisfied
  and chunks correspond to clauses.

  Now we will show that this solution has cost $\Thr$.  Due to the
  Bandwidth Lemma (Lemma~\ref{lem:bandwidth-lemma}), we only have to
  consider the communication cost. We sum inner-gadget communication
  and communication among gadgets to get exactly $\Thr$.

  ($\Leftarrow$) Let us take any solution to $\Problem$ constructed
  based on $\Formula$ of cost $\leq \Thr$.  We will construct a
  positive valuation $\Val$ by considering the nodes in the solution
  to $\Problem$.

  We make the following observations. In every solution of cost
  $\leq \Thr$, every gadget has exactly $\clauses$ many nodes at its
  leaves. This is due to the Bandwidth Lemma
  (Lemma~\ref{lem:bandwidth-lemma}).  Also, inside every gadget either
  all nodes are in the $\positive(\variab)$ subtree of variable
  $\variab$, or in the $\negative(\variab)$ subtree. This is true
  because the cost of a solution where at least one gadget has nodes
  distributed across subtrees is always greater than $\Thr$.

  Now we can construct our valuation $\Val$, as follows (for each
  variable $\variab$ in $\Formula$): If $v_1$ hosts a node then
  $\Val(\variab) = 1$, otherwise $\Val(\variab) = 0$.

  The valuation $\Val$ satisfies all clauses, and hence $\Formula$, as
  the solution to $\Problem$ covers all chunks. To see this, consider
  the leaf which hosts a node which is assigned to any given chunk
  (i.e., the leaf handling any given clause chunk); it is a witness
  that the corresponding clause is satisfied.

\end{proof}

We conclude by observing that our construction leverages the fact that
the number of nodes may exceed the number of chunk types, e.g., for a
clause $(x \vee y \vee z)$ in $\Formula$, both $x$ and $y$ being true
implies the mapping of nodes on vertices labeled $x_1$ and $y_1$, and
which contain the same chunk $c_1$.

\subsection{Hardness of problem with two replicas of each type}\label{sec:two}

We can see that proof from previous section can be carried from
3SAT. This way we need only three replicas of each chunk
type. 2SAT is not NP-hard, not allowing to carry previous
construction for two replicas of each type. In this section we will
show how to modify the construction to show NP-hardness of problem
constrained to have at most two replicas of each type.

Our results so far indicate that dealing with replication can be
challenging.  However, all our hardness proofs concerned scenarios
with three replicas, which raises the question whether the problems
are solvable in polynomial time with a replication factor of
$2$. (Similarly to, say, the $\ZSAT$ problem which is tractable in
contrast to $\TSAT$.)

In the following, we show that this is not the case: the problem
remains NP-hard, at least in the capacitated network.

The proof is by reduction from $\TSAT$. Given a formula $\Formula$ in
conjunctive normal form, consisting of $\clauses$ clauses and $\vars$
variables, we construct a problem instance and substrate tree
$T_{\Formula}$ using two types of gadgets: gadgets for variables and
gadgets for clauses. \emph{Nota bene:} unlike in the previous proof we
will create three chunk types instead of just one, for every clause.

\begin{figure}[htbp]
  \includegraphics{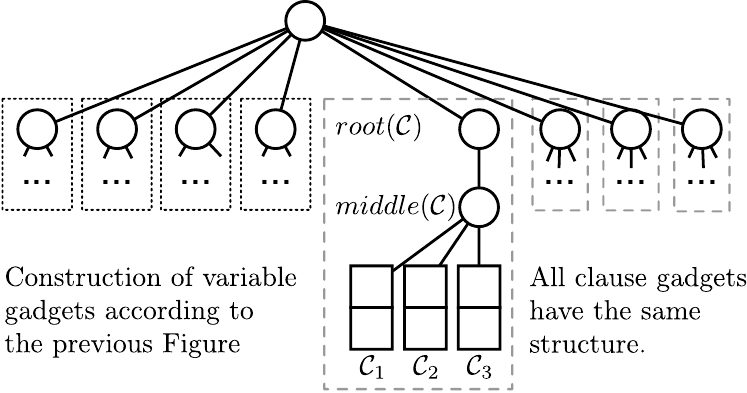}
  \caption{Structure of clause gadgets.}
  \label{fig:clause-gadget}
\end{figure}

\textbf{Construction.}  We build upon the construction for variable
gadgets introduced (see also Figure~\ref{fig:construction_3sat}).

\begin{enumerate}
\item \emph{Tree Construction}: In addition to the variable gadgets
  known from the previous construction, we introduce \emph{clause
    gadgets}. The clause gadget for a clause $\mathcal{C}$
  (illustrated in Figure~\ref{fig:clause-gadget}) has two inner
  vertices: \emph{root}$(\mathcal{C})$,
  \emph{middle}$(\mathcal{C})$\footnote{The only purpose of the middle
    vertex is to maintain the balanced tree property.} and three
  leaves $\mathcal{C}_1,\mathcal{C}_2$ and $\mathcal{C}_3$ . We
  connect leaves to the middle vertex, and the middle vertex to
  \emph{root}$(c)$. We attach the gadget to the tree by linking
  directly the global root to \emph{root}$(\mathcal{C})$. We construct
  our tree out of $\vars$ variable gadgets and $\clauses$ clause
  gadgets.
\item \emph{Chunk Distribution:} For each clause $\mathcal{C}$, we
  generate $3$ chunks types with $2$ replicas each. Each server in the
  clause gadget of $\mathcal{C}$ holds a replica of a different chunk
  type. The remaining replicas of chunk types, are placed in the
  variable gadgets of the variables which satisfy the clause,
  similarly to the previous proof.
  % \maciek{emphasise that those 3 chunks are placed in place of
  % original one chunk} }
  Thus, in total, $6 \cdot \clauses$ variable chunks are distributed
  in the substrate network.  We will consider a setting where
  $\clauses \cdot \vars + 2\clauses$ nodes need to be mapped. Our
  intention is that in every variable gadget, there will be $\clauses$
  nodes, and in every clause gadgets there will be two nodes.

\item \emph{Bandwidth Constraints:} The available bandwidth of the top
  edge of the gadget of each variable $\variab$ is set to
  $\capa(\variab) = \clauses (\clauses(\variables -1) + 2 \cdot
  \clauses)$.
  This value, results from $\clauses$ nodes in the gadget for
  $\variab$, which each have to communicate to
  $\clauses \cdot (\variables -1)$ nodes in other variable gadgets and
  $2\cdot \clauses$ nodes in clause gadgets.  The available bandwidth
  for the top edge of each clause gadget is set to
  $\capa(\clauses) = 2 (\clauses \cdot \variables + 2(\clauses -1)) $.
  This value allows both of the $2$ nodes in a clause gadget, to
  communicate to the $\clauses \cdot \variables$ nodes in variable
  gadgets and the the other $2(\clauses -1)$ nodes in clause gadgets.

\item \emph{Additional Properties:} We set the threshold in a similar
  fashion as in previous proofs. That is, the threshold depends on the
  intra-clause communication cost ($2$ hops), and the inter-clause
  communication cost ($6$ hops). We set the number of nodes to be
  placed to $\clauses \cdot \vars + 2 \cdot \clauses$. We set the
  hosting capacity of each server to $1$, and set
  $\CostTrans = \Thr + 1$ to disallow remote chunk access. We set
  $\CostCom = 1$.

\end{enumerate}

\textbf{Proof of correctness.}  We first prove the following helper
lemma.
\begin{lemma}
  Every valid solution to $\ProblemT$ with cost at most $\Thr$ has the
  property that there are exactly $\clauses$ nodes in each of the
  $\vars$ variable gadgets and exactly two nodes in each of the
  $\clauses$ clause gadgets.
\end{lemma}

Correctness follows from the extended bandwidth lemma
(Lemma~\ref{lem:bandwidth-lemma-extended}).

\begin{theorem}
  $\ProblemT$ is NP-hard.
\end{theorem}
\begin{proof}
  We show that $\ProblemT$ has a solution of cost $\leq \Thr$ if and
  only if $\Formula\in \TSAT$ is satisfiable.

  ($\Rightarrow$) If we have a positive valuation of $\Formula$, we
  fill variable gadgets with nodes like in the proofs before. Then we
  place $2$ nodes in each of the $\clauses$ clause gadgets as follows:
  Given a clause $\mathcal{C} = \ell_1 \lor \ell_2 \lor \ell_3$, we
  pick an arbitrary literal which satisfies the clause. Subsequently
  we place nodes at the leaf nodes in the clause gadget, which
  correspond to the other two literals.  This strategy ensures that
  all chunks can be assigned to collocated nodes, as the only chunk
  type, which cannot be assigned to a collocated node in the clause
  gadget, has a node collocated with its second replica in the
  variable gadget.

  We will then assign chunks to nodes in the following way: For chunk
  type we assign the replica in the variable gadgets to a collocated
  node. If this node does not exist, we assign the replica in the
  clause gadgets, to its collocated node.

  Thus, we have produced a feasible solution of cost $\Thr$.
  ($\Leftarrow$) Let us take any solution $\Sol$ to $\ProblemT$ of
  cost $\leq \Thr$.  Similar to the proof of
  Theorem~\ref{theorem:3sat} and Lemma~\ref{lem:bandwidth-lemma} all
  nodes which are placed in a variable gadgets, will be located in
  either the \emph{positive} or the \emph{negative} subtree. Then we
  can compute a positive valuation by setting each variable $\variab$
  as follows:

  $\Val(\variab) =
  \begin{cases}
    1 & \mbox{iff there is a node at the first leaf}\\
    & \mbox{on positive side of $\variab$ gadget in $\Sol$}\\
    0 & \mbox{otherwise}
  \end{cases}$

  The theorem now follows from the following two additional lemmas.
  \begin{lemma}
    For every clause there exists a node in a variable gadget that
    processes one of three chunks that correspond to that clause.
  \end{lemma}
  \begin{proof}
    Each of the three chunks that correspond to each clause, is
    assigned a collocated node.  At least one of those three nodes is
    not idle in a variable gadget; otherwise, those two nodes in the
    clause gadgets would not suffice in satisfying all chunk types.
  \end{proof}

  Observe that it might happen that in $\Sol$, two nodes in clause
  variables are idle, and three nodes in variable gadgets are
  processing those $3$ chunk types. In this case, arbitrary nodes can
  be taken for the rest of the proof.

\begin{lemma}
  $\Val$ satisfies $\Formula$.
\end{lemma}
\begin{proof}
  Let us consider the matching $M$ of $\Sol$, and let us consider an
  arbitrary clause of $\Formula$ as well as its three chunk types: Due
  to bandwidth constraints, at most two of the chunks types, can be
  processed by nodes in the clause gadgets. We identify any chunk
  type, which is not assigned to a replica in the clause gadgets. The
  processed replica of that chunk type was located in a variable
  gadget. Depending on whether the replica was located in the positive
  or the negative subtree, we set the value of the according variable
  to $1$ (positive subtree) or $0$ (negative subtree).
\end{proof}

\end{proof}
\section{The Bandwidth Lemmas}\label{sec:bw}

\begin{lemma}[Bandwidth Lemma]\label{lem:bandwidth-lemma}
  Let $\clauses$ and $\variables > 4$ be two arbitrary positive integers. Let $a_1, a_2, \ldots,
  a_{\variables}$ be a sequence of $\variables$ integers which adds up to $\clauses \cdot \variables$. Also, for
  each $i$ we have $a_i \leq 2 \cdot \clauses$. Then it holds that if
  $$ \forall_i:~~ a_i \cdot (\clauses \cdot \variables - a_i) \leq \clauses \cdot (\clauses \cdot \variables -
  \clauses), $$
\noindent  then for each $i$: $a_i = \clauses$.
\end{lemma}
\begin{proof}
  By contradiction. Let us assume
that there exists an index $k$ such that
$a_k \neq \clauses$. Then we can distinguish between two cases:
either $a_k<\clauses$ or
$a_k>\clauses$.

\textbf{Case $a_k<\clauses$:} If there exists a $k$ with $a_k<\clauses$,
due to the fact that the sequence adds up to $\clauses \cdot \variables$,
there must also exist a $k'$ such that $a_{k'}<\clauses$ (by a simple
pigeon hole principle). Thus, this case can
also be reduced to the second case (Case $a_k>\clauses$) proved
next.

\textbf{Case $a_k>\clauses$:} Since it also holds that $a_k < 2\clauses$,
$a_k$ must be of the form $\clauses + x$ for $x \in [1, \ldots, \clauses]$.
Let us consider the (bandwidth) inequality:
$$ (\clauses + x) \cdot (\clauses \cdot \variables - \clauses - x) \leq \clauses \cdot (\clauses \cdot \variables - x) $$

This can be transformed to:

$$ 0 \leq x(x - (\clauses \cdot (\variables - 2))) $$

The equation holds for $x \leq 0$ or $x \geq \clauses \cdot (\variables - 2)$,
and no
positive $x \leq \clauses$ can satisfy this inequality for $\variables > 4$. Contradiction.
\end{proof}
\begin{lemma}[Extended Bandwidth Lemma]\label{lem:bandwidth-lemma-extended}
%\maciek{this lemma might be expressed more clearly if we use two
%  sequences instead of one - now $a_1, \ldots, a_{\variables}$
%  correspond to number of VMs in variable gadgets and remaining of the
%  sequence corresponds to number of VMs in clause gadgets.}
  Let $\clauses$ and $\variables > 4$ be two arbitrary positive integers. Let $a_1, a_2, \ldots,
  a_{\clauses}$ and $b_1, b_2, \ldots,
  b_{\variables}$ be two sequences of integers (numbers of nodes in clause and 
in variable gadgets). The sum of all elements in $a$ and $b$ adds up to
  $\clauses \cdot \variables + \clauses \cdot 2$ (number of nodes). Also
  we have $a_i \leq 2 \cdot \clauses$
  (variable gadget node hosting capacity -- equal to number of leaves),
  and $b_i \leq 3$ (clause gadget node hosting capacity). If uplink of variable 
gadget does not exceed bandwidth constraints
  $$ \forall_{i\leq\variables}:~~ b_i \cdot (\clauses \cdot \variables
  + 2\cdot \clauses- b_i) \leq \clauses \cdot (\clauses \cdot \variables -
  \clauses + 2 \cdot \clauses), $$
and uplink of clause gadget does not exceed bandwidth constraints
$$ \forall_{i\leq\clauses}:~~ a_i \cdot (\clauses \cdot \variables + 2 \cdot \clauses - a_i) \leq 2 \cdot (\clauses \cdot \variables -
  2 \cdot \clauses - 2), $$

\noindent  then for each $i\leq \variables$: $b_i = \clauses$ and for
each $i\leq\clauses$: $a_i = 2$ (we have expected number of nodes in variable 
and clause gadgets).
\end{lemma}

We can prove the extended bandwidth lemma by pidgeon hole
principle. However, easier way exists. We sum available bandwidth on
all uplinks of clause gadgets to $C$ and bandwidth an uplinks of
variable gadgets to $V$. The only way that we can distribute nodes
between clause and variable gadgets is to have $2 \cdot \clauses$ in
total in clause gadgets and $\clauses \cdot \variables$ in variable gadgets. To conclude, we
 apply bandwidth lemma \ref{lem:bandwidth-lemma} to clause gadgets
and separatly to variable gadgets.

%%%%%%%%%%%%%%%%%%%%%%%%%%%%%%%%%%%%%
\section{Related Work}\label{sec:relwork}

There has recently been much interest in programming models and distributed
system architectures for the processing and analysis of big data (e.g.~\cite{nodb,mapreduce,shark}). The model studied in
this paper is motivated by MapReduce~\cite{mapreduce} like batch-processing applications, also known
from the popular open-source implementation \emph{Apache Hadoop}.
These applications
generate large amounts of network traffic~\cite{orchestra,talk-about,amazonbw},
and over the last years, several systems have been proposed which provide
a provable network performance, also in shared cloud environments, by supporting
relative~\cite{faircloud,elasticswitch,seawall}
or, as in the case of our paper, \emph{absolute}~\cite{oktopus,secondnet,drl,gatekeeper,proteus} bandwidth reservations
between the virtual machines.
%In particular, the notion of virtual networks which combine compute and network resources has been introduced.
%For a good survey on network virtualization and in particular virtual network embeddings,
%we refer the reader to~\cite{boutaba-survey} and~\cite{fischer-survey}.

The most popular virtual network abstraction for batch-processing applications today is the \emph{virtual cluster},
introduced in the Oktopus paper~\cite{oktopus}, and later studied by many others~\cite{talk-about,proteus}.
Several heuristics have been developed to compute ``good'' embeddings of virtual clusters: embeddings
with small footprints (minimal bandwidth reservation costs)~\cite{oktopus,talk-about,proteus}.
The virtual network embedding problem has also been studied for more general graph abstractions
(e.g., motivated by wide-area networks).~\cite{boutaba-survey,fischer-survey}

%~\cite{infocom2009,ammar,turner,simannealing,zhu06}.

%%%%%%%%%%%%%%%%%%%%%%%%%%%%%%%%%%%%%
\section{Summary and Conclusion}\label{sec:conclusion}
We shown that several embedding problems are NP-hard already in 
 three-level trees---a practically relevant result given today's datacenter topologies~\cite{fattree})---and even if the the number of replicas is bounded by two.

%%%%%%%%%%%%%%%%%%%%%%%%%%%%%%%%%%%%%
%\bibliographystyle{alpha}
{\footnotesize \renewcommand{\baselinestretch}{.9}
\bibliographystyle{abbrv}
\bibliography{references}
}

\end{document}